\documentclass[submission,copyright,creativecommons]{eptcs}
\providecommand{\event}{NCL2026} 

\usepackage{iftex}
\usepackage{graphicx}
\usepackage{setspace}
\usepackage{amsmath,amssymb}
\usepackage{proof, amssymb, my-macro2}
\usepackage{amsthm}
\usepackage{comment}
\usepackage{upgreek}

\ifpdf
  \usepackage{underscore}         
  \usepackage[T1]{fontenc}        
\else
  \usepackage{breakurl}           
\fi

\newtheorem{lemma}{Lemma}
\newtheorem{theorem}{Theorem}

\title{A Basic Modal Logic for the Abstract Treatment of Quantum Mixed States}
\author{Tomoaki Kawano
\institute{Faculty of Informatics\\
Kanagawa University\\
Yokohama, Japan}
\email{tomoakikawano.tk@gmail.com}
}
\def\titlerunning{A Basic Modal Logic for the Abstract Treatment of Quantum Mixed States}
\def\authorrunning{T. Kawano}
\begin{document}
\maketitle

\begin{abstract}
Quantum logic has developed in various directions since its inception. However, many existing logical approaches either focus only on pure states, or use Hilbert-space models in which the abstract relationship between pure states, mixed states, and probabilities is not made explicit. In this work, we introduce a model designed to handle quantum mixed states abstractly. Furthermore, we construct a proof system for the proposed logic and establish its completeness with respect to the intended semantics.
\end{abstract}

\section{Introduction}

Quantum logic provides a formal framework for reasoning about propositions arising in quantum mechanics and quantum information theory. In standard Hilbert-space semantics, propositions are represented by closed subspaces or projection operators, and logical operations are interpreted by lattice-theoretic constructions \cite{birk,HANDQ,HANDQS}.

A central feature of quantum theory is its probabilistic structure. In particular, one should distinguish two sources of probability. The first is the probability internal to a pure state, determined by the Born rule. The second is the classical probabilistic mixing of pure states, represented by mixed states. This distinction is fundamental in open quantum systems, statistical ensembles, and quantum information processing.
Various logical systems for probability \cite{FaginHalpernMegiddo,HANDMANY2,Halpern,Nilsson,HANDMANY} and quantum probability \cite{batprob,intq1,MeydenPatra,PB2018,PB2023,extendQL} have been studied extensively. However, many existing approaches either focus only on pure states, treat probabilities without explicitly separating these two levels, or use Hilbert-space models in which the abstract relationship between pure states, mixed states, and probabilities is not made explicit. While concrete treatments based on Hilbert spaces are important, abstract models are also useful for comparing quantum-probabilistic reasoning with other logics of probability and for conducting independent proof-theoretic analyses.

The aim of this study is to develop a finite probabilistic Kripke-style modal abstraction, together with a proof system for it. The finiteness of the framework should be understood as a restriction on the expressive resolution of the object language, not as a claim that quantum probabilities themselves are discrete. We fix a finite set of probability thresholds and reason only about distinctions observable at that level of precision.
Modal and Kripke-style approaches to quantum logic abstract away from the full Hilbert-space formalism while preserving some essential quantum features such as non-classical entailment \cite{Chihilbert,HANDQ,TK1,nis1,nis2}. However, abstract treatments of probabilistic degrees in such Kripke-style quantum-logical settings have been less extensively developed. Standard dynamic quantum logic, for example, represents states as possible worlds and treats quantum state changes, but its basic semantic framework is formulated in terms of certainty rather than probability \cite{bat1,bat15} or contains only the notion of probability for pure states \cite{TKQPL}.

In the proposed models, possible worlds represent pure states, weighted accessibility relations abstract transition probabilities between pure states, and probability distributions over possible worlds represent mixed states. We introduce a bimodal probabilistic language with two modal operators: one for transition probabilities at the level of pure states and one for probabilistic validity with respect to mixed states. Thus, the semantics separates probabilities determined at the level of pure-state transitions from probabilities arising from classical mixtures of pure states. 
The present framework preserves only a limited part of the quantum structure relevant to probability. It does not attempt to reproduce the full Hilbert-space structure: propositions are not required to form an orthomodular lattice, and notions such as phase, unitary dynamics, sequential measurements, and entanglement are not represented in the present language. Thus, this study should be understood as an abstraction of the two-level probabilistic structure of quantum states rather than as a complete logic of Hilbert-space quantum mechanics.


This semantic separation also guides the proof theory. Since the pure-state component is relational, it is naturally handled by nested sequents \cite{Poggiolesi2023,Poggiolesi2009}. By contrast, the mixed-state component concerns distributions over individual states and is naturally handled by hypersequents \cite{AvronHypersequent,Min68}. We therefore develop a sequent calculus {\bf MSQLM} that combines nested sequents and hypersequents in accordance with the two-level structure of the semantics.

The main contributions of this study are as follows. First, we introduce quantum probabilistic models, finite Kripke-style models in which pure states are represented as worlds and mixed states as probability distributions over worlds. Second, we define a bimodal probabilistic language interpreted over these models. Third, we develop the sequent calculus {\bf MSQLM}. Finally, we prove soundness, completeness, cut elimination, and decidability for the proposed system.

Section 2 introduces the semantic framework and the bimodal language. Section 3 presents the sequent calculus and establishes the main metatheoretic results. Section 4 concludes the study. 

\section{Modal logic for quantum mixed states}
In this section, we develop a finite Kripke-style model and logic for representing two levels of probability inspired by quantum pure and mixed states.
We shall refer to the new logic as QLM (quantum logic with mixed state).
Intuitively, this logic combines the logic for probability in pure states discussed in \cite{TKMB,extendQL} with the logic for classical probability, which is based on a model that assigns probabilities to each possible world in a Kripke model.

We now give concrete examples of probabilities associated with pure and mixed states in quantum mechanics.
A pure state is represented by a unit vector in a Hilbert space, and measurement probabilities are determined by the Born rule.
If a proposition is represented by a closed subspace $S$, then the probability that it holds in a pure state ${\uppsi}$ is given by $\| P_S \uppsi \|^2$, where $P_S$ is the projection onto $S$.
For example, in a two-dimensional complex Hilbert space, the states
$|0\rangle$ and $(1/\sqrt{2}) (|0\rangle + |1\rangle)$
(or their density-matrix representations
$|0\rangle\langle 0| $ and $ (1/\sqrt{2}) (|0\rangle + |1\rangle) (1/\sqrt{2})(\langle 0| + \langle 1|)$) are pure states.
When a measurement in the basis is performed, the former state yields 
$0$ with certainty, whereas in the latter state, the probabilities of obtaining 0 or 1 are each one-half : $ | \langle 0| (1/\sqrt{2}) (|0\rangle + |1\rangle) |^2 = | \langle 1|  (1/\sqrt{2})(|0\rangle + |1\rangle) |^2 = (1/\sqrt{2})^2 = 1/2$.
A mixed state is defined as a classical probabilistic mixture of such pure states.
For example, a probabilistic mixture of these two states with probabilities $1/3$ and $2/3$, respectively, is expressed as the density matrix $ \rho = (1/3)|0\rangle\langle 0| + (2/3)  (1/\sqrt{2}) (|0\rangle + |1\rangle) (1/\sqrt{2})(\langle 0| + \langle 1|)$.
In this example, the probability of obtaining 0 is $(1/3)\times 1 + (2/3) \times (1/\sqrt{2})^2 = 2/3$.
As explained below, when $p$ is understood as the proposition ``the outcome 0 is obtained,'' the formula $\Diamond_1 p$ holds at $|0\rangle$, $\Diamond_{1/2} p$ holds at $(1/\sqrt{2}) (|0\rangle + |1\rangle)$, and ${\sf M_{2/3}} p$ holds at $\rho$.
This example is meant only to motivate the two-level structure of QLM: the probabilities inside each pure state are represented by weighted accessibility, while the outer classical mixture is represented by a probability distribution over worlds.


Propositional variables are denoted by $p,q,r,\ldots$, {\it formulas for pure states} by 
$A, B,\ldots$, and {\it formulas for mixed states} by 
$\phi, \psi,\ldots$.
In this study, we adopt a discretization method that divides the unit interval of probability values into finitely many subintervals.
Let $N_I$ and $N_J$ be positive integers.
Then $I$, $J$, and $I_J$ are defined as follows.
\begin{description}
\item $I=\left\{\frac{k}{N_I}\mid k=0,1,\ldots,N_I\right\},$
\item $J=\left\{\frac{k}{N_J}\mid k=0,1,\ldots,N_J\right\},$
\item $I_J = \left\{\frac{k}{N_I N_J}\mid k=0,1,\ldots,N_I N_J\right\},$  
\end{description}
For example, if $N_I = 10$ and $N_J = 2$, then  
\begin{description}
\item $I=\{0, 0.1, 0.2, \ldots,1\}$,
\item $J=\{0, 0.5, 1\}$,
\item $I_J=\{0, 0.05, 0.1, 0.15, \ldots,1\}$. 
\end{description}
The set $J$ is used for transition weights at the pure-state level, while $I$ is used for probability weights in mixed-state distributions. Since mixed-state probabilities combine these two kinds of weights, $I_J$ is used as the threshold set for mixed-state formulas. 

Our use of finite probability thresholds follows two related traditions. Suppes-style probability logics use a fixed error parameter $\epsilon$ and reason with assertions of the form ``probability at least $1-n \epsilon$'' \cite{Boricic1}. Finitely graded modal approaches to extended quantum logic similarly restrict possible degrees to a finite threshold set in order to obtain a tractable proof system \cite{TKMB,TKQPL,extendQL}. Thus, the restriction to finite thresholds concerns the expressive resolution of the object language, not the underlying notion of probability itself. Probabilities lying between two adjacent thresholds in $I_J$ are not distinguished by the language. Accordingly, QLM is a finite-threshold modal abstraction of mixed-state reasoning, rather than a full real-valued logic of Hilbert-space quantum mechanics.
Allowing arbitrary real-valued thresholds would not present a conceptual difficulty for the semantics itself. However, the present proof system in this study relies essentially on the finiteness of $I$ and $J$: only finitely many modal formulas and probability configurations need to be considered during proof search. With continuously many thresholds, the present finitary completeness and decidability arguments would therefore no longer apply directly.


We define the formulas of QLM as follows:

\begin{description}
\item $A :: =  p \ | \ \neg A \ |\  A \wedge A \ | \ \square_\gamma A $
\item $\phi :: =  A \ | \ \neg \phi \ |\  \phi \wedge \phi \ | \ {\sf M_\alpha} A$         
\end{description}

\noi
where $\gamma \in J \setminus \{0\}$, $\alpha \in I_J \setminus \{0\}$. 
(The formulas for pure states are based on \cite{TKMB,extendQL}; however, for simplicity, we use only the main ones in this study.)
For the sake of readability, we sometimes write $(\gamma,i)$ instead of $\gamma_i$.
Let $m$ be the minimal element of $J \setminus \{0\}$ (or, depending on the context, 
the minimal element of $I_J \setminus \{0\}$).
We use the following abbreviations: $\Diamond_\gamma A = \neg \square_{\gamma} \neg A$, and $\sim A = \square_m \neg A$. 

We define a {\it quantum probabilistic model} (QPM) $\langle W, R, V, P \rangle $ as follows:

\begin{description}
\item $W$ is a set (possible worlds),
\item $R$ is a function from $W \times W$ to $J$ such that, $R(x, y) = 1$ iff $x = y$ (reflexivity), and $R(x, y) = R(y, x)$ (symmetry),
\item $V$ is a function from propositional variables to the power set of $W$,
\item $P$ is a function from $W$ to $I$ such that $\sum_{x \in W} P(x) = 1$.

\end{description}

Intuitively, the elements of $W$ represent pure states (rays, or unit vectors modulo global phase), and 
$R$ represents the squared absolute value of the inner product between states.
That is, $R$ is related to the probability in quantum pure states, whereas $P$ corresponds to classical probabilistic mixtures.

We define the truth of formulas at $x \in W$ as follows:

\begin{description}
\item $x \models p\ \Leftrightarrow x \in V(p)$,

\item $x \models \neg A \ \Leftrightarrow x \not\models A$,

\item $x \models A \wedge B \ \Leftrightarrow x \models A$ and $x \models B$,

\item $x \models \square_\gamma A \ \Leftrightarrow $
$\forall y \in W$,  $\gamma \leqq R (x, y) \to y \models A $,

\item $x \models {\sf M_\alpha} A \ \Leftrightarrow$ $ \alpha \leqq \sum_{y \in W} (P(y)*\max ( \{ 0 \} \cup \{\gamma \in J \setminus \{0\} |\ y \models \Diamond_\gamma A \}))$.
\end{description}


Note that no rounding is required. If $P(y) = k_y/N_I$ and $q_A(y) = \max ( \{ 0 \} \cup \{\gamma \in J \setminus \{0\} |\ y \models \Diamond_\gamma A \})) = l_y/N_J$, then $\sum_y P(y) q_{A}(y) = \sum_y k_y l_y / N_I N_J \in I_j$. 

The normalized projection of $x$ onto the subspace in which 
$A$ is true can be understood as the unit vector in that subspace that is closest to 
$x$.
From the definition of $x \models \square_\gamma A$, we can confirm that 
$\Diamond_\gamma A$ expresses the existence of at least one 
$A$-true state closer than 
$\gamma$, which in turn means that, in this pure state, $A$ is obtained with probability at least $\gamma$.
(Since $\gamma$ is related to the inner product, ``closer'' here means that 
$R (x, y)$ is larger.)
Therefore, ${\sf M_\alpha} A$ means that, in the present mixed state, 
$A$ is obtained with probability at least $\alpha$.
$\square_m \neg A$ states that it is orthogonal (i.e., the inner product is zero) to any state in which 
$A$ holds.
Therefore, $\sim$ denotes {\it quantum negation} and expresses that the probability of obtaining $A$ being true is zero.

We disallow nested {\sf M}-modalities, since iterating probability claims about the same current mixed state has little clear meaning and would only complicate the system.
Intuitively, this is analogous to the expression of probability in \cite{Boricic1}, where probability enters only through sequents, 
while the formulas inside the sequents remain classical propositional formulas.
The inclusion of $A$ in the grammar of mixed-state formulas is only a syntactic embedding of pure-state formulas into the mixed-state language. The genuinely mixed-state probabilistic operator is ${\sf M}_\alpha$, which applies only to pure-state formulas.
The definition is separated into two steps only for this reason, and the symbols $\neg$ and $\wedge$ appear in both grammars, but they denote the same Boolean connectives.
We also omit the modal operator corresponding to ``at least $0$'' from the definition, since it is of little probabilistic significance and would unnecessarily complicate the discussion as a global operator.




\ 

\noi
Example 1 of QPM:

\begin{description}
\item $W = \{x, y, z, w\}$, $R(x,z) = R(x,w) = R(y,z) = R(y,w) = 0.5$, $R(x,y) = R(z,w) = 0$, $V(p) = \{x\}, \\ V(q) = \{ y\}, V(r) = \{z\}, V(s) = \{ w\}$, $P(x) = P(y) = 0.5, P(w) = P(z) = 0$.
\end{description}

\noi
Example 2 of QPM:

\begin{description}
\item $W = \{x, y, z, w\}$, $R(x,z) = R(x,w) = R(y,z) = R(y,w) = 0.5$, $R(x,y) = R(z,w) = 0$, $V(p) = \{x\}, \\ V(q) = \{ y\}, V(r) = \{z\}, V(s) = \{ w\}$, $P(x) = P(y) = 0, P(w) = P(z) = 0.5$.
\end{description}

\noi
In both models, $W$ (and $p,q,r,s$) represent the four fundamental quantum bit states $| 0 \rangle$,
$| 1 \rangle$, $(1/\sqrt{2})(| 0 \rangle + | 1 \rangle)$, and $(1/\sqrt{2})(| 0 \rangle - | 1 \rangle)$. 
Example 1 represents an equal probabilistic mixture of 
$| 0 \rangle$ and $| 1 \rangle$, while Example 2 represents an equal probabilistic mixture of $(1/\sqrt{2})(| 0 \rangle + | 1 \rangle)$ and $(1/\sqrt{2})(| 0 \rangle - | 1 \rangle)$.
As is well known, these two are identical as mixed states:
$(1/2)(|0 \rangle \langle 0|)+(1/2)(|1 \rangle \langle 1|) = 
(1/2)(1/\sqrt{2})(|0 \rangle + |1 \rangle)(1/\sqrt{2})(\langle 0| + \langle 1|) +(1/2)(1/\sqrt{2})(|0 \rangle - |1 \rangle)(1/\sqrt{2})(\langle 0| - \langle 1|)$.
This notation is reflected in the models. That is, for all $\alpha$, the truth of ${\sf M}_{\alpha}p$, ${\sf M}_{\alpha}q$, ${\sf M}_{\alpha}r$, and ${\sf M}_{\alpha}s$ coincides in both models.
(${\sf M}_{0.5}t$ is true but ${\sf M}_{0.6}t$ is false for all $t \in \{p,q,r,s\}$.)
QLM can express threshold statements concerning measurement probabilities in a mixed ensemble, such as ``the probability of obtaining outcome $p$ is at least 2/3,'' while simultaneously referring to transition probabilities at each underlying pure state. On the other hand, the present language cannot directly express entanglement between subsystems, unitary time evolution, or sequential measurements.

We refer to the version of QPM with the following restrictions as $N_I$-QPM.

\begin{description}
\item $|W| \geqq N_I$,

\item For all $x \in W$, $P(x) = 1/N_I$ or $P(x) = 0$.

\end{description}
Mixed states that are not uniform mixtures can also be represented in 
$N_I$-QPM.
For example, a mixed state in which $| 0 \rangle$ has probability 1/4 and 
$| 1 \rangle$ has probability 3/4 can be represented by the following method, which duplicates the ``frame''. (If the values of $R$ between possible worlds in different frames are set to zero, they do not affect the truth of formulas of $\square_{\gamma}$ in each frame.)

\begin{description}
\item $W = \{x_1, y_1, z_1, w_1, x_2, y_2, z_2, w_2, x_3, y_3, z_3, w_3\}$
\item $R(x_i,z_i) = R(x_i,w_i) = R(y_i,z_i) = R(y_i,w_i) = 0.5$. Otherwise, if $u 
\neq v$, $R(u,v) = 0$
\item $V(p) = \{ x_1, x_2, x_3\}, V(q) = \{ y_1, y_2, y_3\}, V(r) = \{z_1, z_2, z_3\}, V(s) = \{ w_1, w_2, w_3\}$
\item $P(x_1) = P(y_1) = P(y_2) = P(y_3) = 0.25$. Otherwise, $P(v) = 0$
\end{description}

\noi
Therefore, the choice of QPM or $N_I$-QPM is arbitrary and can be adapted to suit the intended purpose.
From this point on, we restrict our attention to $N_I$-QPM.

We define that $\phi$ is valid in QLM (for $I,J$) if $\phi$ is true at all $x \in W$ such that $P(x) = 1/N_I$ of all $N_I$-QPM.
The point with $P(x) = 1/N_I$ represents an abstract pure state within the mixed state.
Therefore, this notion of validity can be understood as being grounded in truth values evaluated in actual pure states. 
We give some examples of valid formulas below.

\begin{description}
\item ${\sf M}_{1} A \to \Diamond_{1} A$ \ \ (If $A$ holds with certainty in the sense of mixed states, then $A$ also holds with certainty in the pure state under consideration.)

\item $A \to {\sf M}_{m} A$ \ \ (In a mixed state that contains a pure state in which $A$ is true, the probability that $A$ holds is nonzero.)

\item ${\sf M}_{\alpha} \sim A \to \neg {\sf M}_{1} A$ \ (for all $\alpha \in I_J \setminus \{0\}$ )
\ \ (If a proposition contrary to $A$ holds with nonzero probability, then $A$ cannot hold with probability $1$.)

\end{description}

\section{Sequent calculus for QLM}

We now construct a proof system for QLM. Several proof-theoretic methods have been proposed for incorporating probabilities into deductive systems. One common method is to attach numerical labels to sequents. For example, as discussed in the previous section, \cite{Boricic1} uses sequents of the form  $\GA \Rightarrow^n \DE$, with the intended meaning that
the probability that $\bigvee \DE$ given that $\bigwedge \GA$, is at least $1-n\epsilon$.
Such labelled sequents are useful for probabilistic reasoning over an underlying classical logic. However, they are not sufficient for the present setting, where the proof system must simultaneously represent two different structures: the relational structure of pure-state transitions and the distributional structure of mixed states. We therefore combine two proof-theoretic devices. Nested sequents are used to represent the pure-state relational component, following \cite{TKMB}, while hypersequents are used to represent the mixed-state distributional component. We begin by recalling the nested sequents for the pure-state fragment.

A {\it sequent} is defined as a structure $\GA \Rightarrow \DE$ consisting of two sets of formulas, 
$\GA$ and $\DE$, and an arrow symbol 
$\Rightarrow$.
$\GA \Rightarrow \DE$ is {\it false} at $x$ if all $A \in \GA$ are true at $x$
and all $A \in \DE$ are false at $x$.
If $\GA \Rightarrow \DE$ is not false at $x$, then it is {\it true} at $x$.
Nested sequents are defined inductively as follows:

\begin{enumerate}
\item A sequent is a nested sequent.
\item If $\GA \Rightarrow \DE$ is a sequent and
$\mathcal{T}$ is a finite set of nested sequents enclosed by numerical brackets $[\ ]_{\gamma}\ (\gamma \in J \setminus \{0,1\})$, then $\GA \Rightarrow \DE, \mathcal{T}$ is a nested sequent.
\item These are the only objects that qualify as nested sequents.
\end{enumerate}

\noi
For example, if $J = \{ 0,0.1,0.2,\ldots,1 \}$, then
$p,\square_{0.8} r \Rightarrow q, [ \Rightarrow r, [p \wedge \square_{0.1} q,r \Rightarrow q]_{0.5}]_{0.2}, [r \Rightarrow \neg p]_{0.7} $
is a nested sequent. (If we denote this nested sequent by $\GA \Rightarrow \DE, \mathcal{T}$, then $\GA \Rightarrow \DE = p,\square_{0.8} r \Rightarrow q$, and $\mathcal{T} = \{ [ \Rightarrow r, [p \wedge \square_{0.1} q,r \Rightarrow q]_{0.5}]_{0.2}, [r \Rightarrow \neg p]_{0.7} \}$.)
Since a nested sequent has a hierarchical structure represented by numerical brackets, it can be regarded as a {\it tree} with weighted edges.
From this perspective, each sequent corresponds to a {\it node}.
In the above example, $(p,\square_{0.8} r \Rightarrow q)$, $(\Rightarrow r)$, $(p \wedge \square_{0.1} q,r \Rightarrow q)$, and $(r \Rightarrow \neg p)$ are nodes and $(p,\square_{0.8} r \Rightarrow q)$ is the {\it root}.

The meaning of a nested sequent is determined by the translation $T$ into a formula.

\begin{description}

\item $T (\GA \Rightarrow \DE) = \bigwedge \GA \to \bigvee \DE$.

\item $T (\GA \Rightarrow \DE, [\GA_1 \Rightarrow \DE_1, \mathcal{T}_1]_{(\gamma, 1)}, \ldots,  [\GA_n \Rightarrow \DE_n, \mathcal{T}_n]_{(\gamma, n)} )$

$= T (\GA \Rightarrow \DE) \vee  \square_{(\gamma, 1)} T (\GA_1 \Rightarrow \DE_1, \mathcal{T}_1) \vee \ldots \vee \square_{(\gamma, n)} T (\GA_n \Rightarrow \DE_n, \mathcal{T}_n) $.

\end{description}

We now extend these notions within the framework of hypersequents.
We define {\it ($N$-fold) multi-level nested sequent} (MN-sequent) as the following structure with the condition; sequents $\GA_i \Rightarrow \DE_i$ may include both
$A,B, \ldots$ and $\phi, \psi, \ldots$ but the formulas appearing in $\mathcal{T}_i$ are restricted to 
$A,B,\ldots$.

$$\GA_1 \Rightarrow \DE_1,\mathcal{T}_1| \GA_2 \Rightarrow \DE_2,\mathcal{T}_2 | \ldots | \GA_{N} \Rightarrow \DE_{N},\mathcal{T}_N$$

\noi
That is, an MN-sequent is a tuple of nested sequents with the condition. 
Treating a nested sequent as a tree structure, 
an MN-sequent may be regarded as an aggregate of several tree structures.
We say that these $\GA_i \Rightarrow \DE_i$ are {\it root sequents}.

The following abbreviations for an MN-sequent are employed.
When writing the full MN-sequent would be redundant, we write 
$G_M$ for $\GA_1 \Rightarrow \DE_1,\mathcal{T}_1| \GA_2 \Rightarrow \DE_2,\mathcal{T}_2 | \ldots | \GA_{N} \Rightarrow \DE_{N},\mathcal{T}_N$.
To indicate that we are focusing on a specific root sequent $\GA_i \Rightarrow \DE_i$ of $G_M$, we write 
MN-sequent $G_M$ as ${\SEQ{G_M \{ | \GA_i}{\DE_i, \mathcal{T}_i | \}}}$.
To indicate that we are focusing on a ``sub-tree'' $\GA \Rightarrow \DE, \mathcal{T}$ (includes the case of the entire tree) of one of the trees in $G_M$, 
we write MN-sequent $G_M$ as
${\SEQ{G_M \{ \GA}{\DE, \mathcal{T} \}}}$.
Hereafter, sequents with the same content are considered distinct if they are different nodes.
For example, in $p,r \Rightarrow q, [ u \Rightarrow r, [p,r \Rightarrow q]]_{0.2}\ |\ q \Rightarrow p $, although $p,r \Rightarrow q$ appears twice, if we denote the first occurrence by $\GA \Rightarrow \DE$, the second occurrence is distinguished and denoted by $\GA' \Rightarrow \DE'$.
If a sequent $\GA \Rightarrow \DE$ is a node of one of the trees in 
$G_M$, we say that 
$\GA \Rightarrow \DE$ {\it occurs} in $G_M$.
If the sequent $\GA \Rightarrow \DE$ occurs in $G_M$, we write $\GA \Rightarrow \DE \sqsubset G_M$. 

A bracketed sequent $[\GA \Rightarrow \DE]_\gamma$ attached to a node represents another pure state whose transition weight from the parent state is at least $\gamma$. Thus, a nested sequent describes a finite tree of possible pure states connected by lower bounds on transition probabilities. An MN-sequent then collects \(N_I\) such trees, corresponding to the pure states occurring with nonzero probability in the mixed state.
We define the {\it embedding} $\mathcal{E}$ of ($N_I$-fold) $G_M$ into 
an $N_I$-QPM $\langle W, R, V, P \rangle $ as follows:

\begin{enumerate}
\item $\mathcal{E}$ is a function from sequents occurring in $G_M$ to $W$.
\item Each root sequent is mapped to $x \in W$ satisfying $P(x) \neq 0$.
\item Each root sequent is mapped to a distinct element. (By the definition of $N_I$-QPM, exactly $N_I$ elements satisfy $P(x) \neq 0$.)
\item If $\GA \Rightarrow \DE$ and $\GA' \Rightarrow \DE'$ are connected by an edge in a tree (that is, if these sequents occur in $G_M$ in the form of 
$\GA \Rightarrow \DE, [ \GA' \Rightarrow \DE',\mathcal{T'}]_{\gamma}, \mathcal{T}$),
and $R(\mathcal{E}(\GA \Rightarrow \DE),\mathcal{E}(\GA' \Rightarrow \DE')) = \gamma'$, 
then $\gamma \leqq \gamma'$.

\end{enumerate}

We define that $G_M$ is false under $\mathcal{E}$ if for all 
$\GA \Rightarrow \DE \sqsubset G_M$, $\GA \Rightarrow \DE$ is false at $\mathcal{E}(\GA \Rightarrow \DE)$.
If $G_M$ is not false under $\mathcal{E}$, then it is true under $\mathcal{E}$.
If $G_M$ is true under any $\mathcal{E}$ for all $N_I$-QPM, then it is valid in $N_I$-QPM.
When it is clear from the context, we simply write
$\mathcal{E}(i)$ for $\mathcal{E}(\GA_i \Rightarrow \DE_i)$.

MN-sequent calculus {\bf MSQLM} is defined as follows:

\noi
Axiom: 

\begin{center} $G_M \{ \phi \Rightarrow \phi, \mathcal{T} \}^{*} $\end{center} 

\noi 
Rules:




$$
\infer[\mbox{(cut)}]{\SEQ{G_M \{  \GA}{\DE, \mathcal{T}  \}}}{
 \SEQ{G_M \{  \GA}{\DE, \phi,  \mathcal{T}  \}} & \SEQ{G_M \{ \phi,  \GA }{\DE , \mathcal{T}   \}}}
$$ 
$$
\infer[\mbox{(wL)$^{*}$}]{\SEQ{G_M \{  \phi,  \GA }{\DE , \mathcal{T}   \}}}
 {\SEQ{G_M \{ \GA }{\DE , \mathcal{T}   \}}}
\hspace{5mm}
\infer[\mbox{(wR)$^{*}$}]{\SEQ{G_M \{  \GA }{\DE , \phi,  \mathcal{T}   \}}}
 {\SEQ{G_M \{  \GA }{\DE , \mathcal{T}   \}}}
$$
$$
\infer[\mbox{($\neg$L)$^{}$}]{\SEQ{G_M \{ \neg \phi,  \GA }{\DE , \mathcal{T}   \}}}
 {\SEQ{G_M \{ \GA }{\DE , \phi,  \mathcal{T}  \}}}
\hspace{5mm}
\infer[\mbox{($\neg$R)$^{}$}]{\SEQ{G_M \{ \GA }{\DE , \neg \phi,  \mathcal{T}  \}}}
 {\SEQ{G_M \{ \phi,  \GA }{\DE , \mathcal{T}  \}}}
$$
$$
\infer[\mbox{(\AND L)$^{}$}]{\SEQ{G_M \{ \phi \wedge \psi, \GA }{\DE , \mathcal{T}  \}}}
 {\SEQ{G_M \{ \phi,  \psi, \GA }{\DE , \mathcal{T}  \}}}
\hspace{5mm}
\infer[\mbox{(\AND R)}]{\SEQ{G_M \{ \GA }{\DE, \phi \wedge \psi, \mathcal{T}  \}}}{
 \SEQ{G_M \{ \GA }{\DE, \phi, \mathcal{T} \}} & \SEQ{G_M \{ \GA }{\DE, \psi, \mathcal{T} \}}}
$$
$$
\infer[\mbox{($\square$ R) $^{}$}]{\SEQ{G_M \{ \GA}{\DE, \square_{\gamma} A, \mathcal{T} \}}}{
 \SEQ{G_M \{ \GA}{\DE, [\Rightarrow A]_{\gamma}, \mathcal{T} \}}}
\hspace{5mm}
\infer[\mbox{($\square$ R self) }]{\SEQ{G_M \{ \GA}{\DE, \square_{1} A, \mathcal{T} \}}}{
 \SEQ{G_M \{ \GA}{\DE, A, \mathcal{T} \}}}
\hspace{5mm}
$$
$$
\infer[\mbox{($\square$ L) $^{**}$}]{\SEQ{G_M \{ \square_{\gamma} A, \GA}{\DE, \ [\GA' \Rightarrow \DE', \mathcal{T'}]_{\gamma'}, \mathcal{T} \}}}{
 \SEQ{G_M \{ \GA}{\DE, [A, \GA' \Rightarrow \DE', \mathcal{T'}]_{\gamma'}, \mathcal{T} \}}}
\hspace{5mm}
\infer[\mbox{($\square$ L self) $^{}$}]{\SEQ{G_M \{ \square_{\gamma} A, \GA}{\DE, \mathcal{T} \}}}{
 \SEQ{G_M \{ A, \GA}{\DE, \mathcal{T} \}}}
$$
$$
\infer[\mbox{($\square$ L sym) $^{**}$}]{\SEQ{G_M \{ \GA}{\DE, \ [\square_{\gamma} A, \GA' \Rightarrow \DE', \mathcal{T'}]_{\gamma'}, \mathcal{T} \}}}{
 \SEQ{G_M\{ A, \GA}{\DE, [\GA' \Rightarrow \DE', \mathcal{T'}]_{\gamma'}, \mathcal{T} \}}}
$$
$$
\infer[\mbox{(${\sf M}$L) $^{}$}]{\SEQ{G_M \{ | {\sf M_\alpha} A , \GA_i}{\DE_i, \mathcal{T}_i| \}}}{
 MSQ_1}
\hspace{5mm}
\infer[\mbox{(${\sf M}$R) $^{}$}]{\SEQ{G_M \{ | \GA_i}{\DE_i, {\sf M_\alpha} A, \mathcal{T}_i | \}}}{
 MSQ_2}
$$
\begin{description}

\item * If the displayed sequent is not the root sequent, then $\phi$ is restricted to $A,B,\ldots$.

\item ** $\gamma \leqq \gamma'$ 


\end{description}

\begin{description}

\item[(1)] $MSQ_1$ is the set of all MN-sequents satisfying the following conditions:

\begin{itemize}


\item An MN-sequent obtained by applying the following procedure to 
$G_M \{ | \GA_i \Rightarrow \DE_i, \mathcal{T}_i| \}$.
For some $i$, we specify $\Diamond_{(\gamma,i)} A$ in the left-hand side of the root of $i$-th nested-sequent.
This $\Diamond_{(\gamma,i)} A$ may be selected from those $\Diamond_{\delta} A$ already contained in $\GA_i$, or it may be introduced as a new element.


\item $\sum( (1/N_I) \times \gamma_i )= \alpha$. \ \ ($\sum$ ranges over those $i$ for which $\Diamond_{(\gamma,i)} A$ is selected.)

\end{itemize}


\item[(2)\label{cMR}]
$MSQ_2$ is the set of all MN-sequents satisfying the following conditions:

\begin{itemize}
\item An MN-sequent obtained by applying the following procedure to 
$G_M \{ | \GA_i \Rightarrow \DE_i, \mathcal{T}_i| \}$.
For some $i$, we specify $\Diamond_{(\gamma,i)} A$ in the right-hand side of the root of $i$-th nested-sequent.
This $\Diamond_{(\gamma,i)} A$ may be selected from those $\Diamond_{\delta} A$ already contained in $\DE_i$, or it may be introduced as a new element.

\item $\sum_{i=1}^{N_I} S_{i} = N_I \times \alpha + (N_I - 1)(1/N_J) $, where
\begin{eqnarray}
S_i =
  \left\{
    \begin{array}{l}
      \gamma_i \ ({\rm if} \ \Diamond_{(\gamma, i)} A {\rm \ is \ specified \ for \ that \ } i) \\
      1 + (1/N_J) \ ({\rm else})
    \end{array}
  \right.
\end{eqnarray}


\end{itemize}

\end{description}

Note that in ({\sf M}L) and ({\sf M}R), ${\sf M}_\alpha A$ can be introduced into an arbitrary $i$-th sequent.
The rules for $\square_{\gamma}$ are extensions of the rules appearing in \cite{TKMB}.
Intuitively, these conditions of (${\sf M}$L) and (${\sf M}$R) capture all possible probability distributions related to ${\sf M}_\alpha A$, thereby ensuring soundness.
In the following example, 
all occurrences of $\Diamond_{(\gamma,i)} A$ are explicitly written for clarity although $\Diamond_{(\gamma,i)} A$ may be included in $\GA_i$ (or $\DE_i$) in some cases.
In the sequent calculus, the premises of a rule are typically displayed horizontally. In the present case, however, they are arranged vertically due to their number.

\ 

\noi
Example 1 of (${\sf M}$L) \ \ \ \ $N_I = 2, N_J = 10$, $\alpha = 0.2$ 

$$
\infer[\mbox{(${\sf M}$L) $^{}$}]{\GA_1 \Rightarrow \DE_1, \mathcal{T}_1 \  |\ {\sf M_{0.2}} A, \GA_2 \Rightarrow \DE_2, \mathcal{T}_2}
{
  \deduce{ \GA_1 \Rightarrow \DE_1, \mathcal{T}_1 \ | \Diamond_{0.4} A, \GA_2 \Rightarrow \DE_2, \mathcal{T}_2}
  {
  \deduce{\Diamond_{0.1} A, \GA_1 \Rightarrow \DE_1, \mathcal{T}_1 \ | \Diamond_{0.3} A, \GA_2 \Rightarrow \DE_2, \mathcal{T}_2}
  {
  \deduce{\Diamond_{0.2} A, \GA_1 \Rightarrow \DE_1, \mathcal{T}_1 \ | \Diamond_{0.2} A, \GA_2 \Rightarrow \DE_2, \mathcal{T}_2}
  {
  \deduce{\Diamond_{0.3} A, \GA_1 \Rightarrow \DE_1, \mathcal{T}_1 \ | \Diamond_{0.1} A, \GA_2 \Rightarrow \DE_2, \mathcal{T}_2}
  {
\Diamond_{0.4} A, \GA_1 \Rightarrow \DE_1, \mathcal{T}_1 \ |\GA_2 \Rightarrow \DE_2, \mathcal{T}_2}}
  }}
}
$$


\noi
Example 1 of (${\sf M}$R) \ \ \ \ $N_I = 2, N_J = 10$, $\alpha = 0.5$, $2 \times 0.5 + (2 - 1)(1/10) = 1.1$

$$
\infer[\mbox{(${\sf M}$R) $^{}$}]{\GA_1 \Rightarrow \DE_1, {\sf M_{0.5}} A, \mathcal{T}_1 \  |  \GA_2 \Rightarrow \DE_2, \mathcal{T}_2}
{ 
  \deduce{\GA_1 \Rightarrow \DE_1, \Diamond_{1} A, \mathcal{T}_1 \ |  \GA_2 \Rightarrow \DE_2,\Diamond_{0.1} A, \mathcal{T}_2}
  {
  \deduce{\GA_1 \Rightarrow \DE_1, \Diamond_{0.9} A, \mathcal{T}_1 \ |  \GA_2 \Rightarrow \DE_2,\Diamond_{0.2} A, \mathcal{T}_2}
  {
  \deduce{\GA_1 \Rightarrow \DE_1, \Diamond_{0.8} A, \mathcal{T}_1 \ |  \GA_2 \Rightarrow \DE_2,\Diamond_{0.3} A, \mathcal{T}_2}
  {
  \deduce{\GA_1 \Rightarrow \DE_1, \Diamond_{0.7} A, \mathcal{T}_1 \ |  \GA_2 \Rightarrow \DE_2,\Diamond_{0.4} A, \mathcal{T}_2}
  {
  \deduce{\GA_1 \Rightarrow \DE_1, \Diamond_{0.6} A, \mathcal{T}_1 \ |  \GA_2 \Rightarrow \DE_2,\Diamond_{0.5} A, \mathcal{T}_2}
  {
  \deduce{\GA_1 \Rightarrow \DE_1, \Diamond_{0.5} A, \mathcal{T}_1 \ |  \GA_2 \Rightarrow \DE_2,\Diamond_{0.6} A, \mathcal{T}_2}
  {
  \deduce{\GA_1 \Rightarrow \DE_1, \Diamond_{0.4} A, \mathcal{T}_1 \ |  \GA_2 \Rightarrow \DE_2,\Diamond_{0.7} A, \mathcal{T}_2}
  {
  \deduce{\GA_1 \Rightarrow \DE_1, \Diamond_{0.3} A, \mathcal{T}_1 \ |  \GA_2 \Rightarrow \DE_2,\Diamond_{0.8} A, \mathcal{T}_2}
  {
  \deduce{\GA_1 \Rightarrow \DE_1, \Diamond_{0.2} A, \mathcal{T}_1 \ |  \GA_2 \Rightarrow \DE_2,\Diamond_{0.9} A, \mathcal{T}_2}
  {
  \GA_1 \Rightarrow \DE_1, \Diamond_{0.1} A, \mathcal{T}_1 \ |  \GA_2 \Rightarrow \DE_2,\Diamond_{1} A, \mathcal{T}_2}
  }}}
  }}}
  }}
}
$$

\noi
Example 2 of (${\sf M}$R) \ \ \ \ $N_I = 3, N_J = 10$, $\alpha = 0.5$, $3 \times 0.5 + (3 - 1)(1/10) = 1.7$

$$
\infer[\mbox{(${\sf M}$R) $^{}$}]{\GA_1 \Rightarrow \DE_1, \mathcal{T}_1 \  |  \GA_2 \Rightarrow \DE_2, {\sf M_{0.5}} A, \mathcal{T}_2 |  \GA_3 \Rightarrow \DE_3, \mathcal{T}_3}
{
  \deduce{\GA_1 \Rightarrow \DE_1, \mathcal{T}_1 \ |  \GA_2 \Rightarrow \DE_2,\Diamond_{0.5} A, \mathcal{T}_2|  \GA_3 \Rightarrow \DE_3,  \Diamond_{0.1} A , \mathcal{T}_3}
  {
  \deduce{\GA_1 \Rightarrow \DE_1, \mathcal{T}_1 \ |  \GA_2 \Rightarrow \DE_2,\Diamond_{0.4} A, \mathcal{T}_2|  \GA_3 \Rightarrow \DE_3,  \Diamond_{0.2} A , \mathcal{T}_3}
  {
  \deduce{
   \vdots}
  {
  \deduce{\GA_1 \Rightarrow \DE_1, \Diamond_{0.3} A, \mathcal{T}_1 \ |  \GA_2 \Rightarrow \DE_2,\Diamond_{1} A, \mathcal{T}_2|  \GA_3 \Rightarrow \DE_3,  \Diamond_{0.4} A , \mathcal{T}_3}
  {
  \deduce{\GA_1 \Rightarrow \DE_1, \Diamond_{0.3} A, \mathcal{T}_1 \ |  \GA_2 \Rightarrow \DE_2,\Diamond_{0.9} A, \mathcal{T}_2|  \GA_3 \Rightarrow \DE_3,  \Diamond_{0.5} A , \mathcal{T}_3}
  {
  \deduce{
   \vdots}
  {
  \deduce{\GA_1 \Rightarrow \DE_1, \Diamond_{0.2} A, \mathcal{T}_1 \ |  \GA_2 \Rightarrow \DE_2,\Diamond_{0.4} A, \mathcal{T}_2| \GA_3 \Rightarrow \DE_3, \mathcal{T}_3}
  {
\GA_1 \Rightarrow \DE_1, \Diamond_{0.1} A, \mathcal{T}_1 \ |  \GA_2 \Rightarrow \DE_2,\Diamond_{0.5} A, \mathcal{T}_2|  \GA_3 \Rightarrow \DE_3, \mathcal{T}_3}}
  }}}
  }}
}
$$

\noi
Example of proof of ${\sf M}_{0.2} \sim A \to \neg {\sf M}_{1} A$

We assume that $I = \{ 0,0.5,1 \}$, and $J = \{ 0,0.1,0.2,\ldots,1 \}$. Therefore, $m=0.1$ and $\sim A = \square_{0.1} \neg A$.

$$
\infer[\mbox{(w L)}]{\SEQ{\neg \square_{1} \neg A, \neg \square_{0.4}  \Diamond_{0.1} A }{| \neg \square_{1} \neg A \Rightarrow }}
{
\infer[\mbox{($\neg$ L)}]{\SEQ{\neg \square_{1} \neg A, \neg \square_{0.4}  \Diamond_{0.1} A }{| \Rightarrow }}
{
\infer[\mbox{($\neg$ L)}]{\SEQ{\neg \square_{1} \neg A }{\square_{0.4}  \Diamond_{0.1} A | \Rightarrow }}
{
\infer[\mbox{($\square$ R self)}]{\SEQ{ }{\square_{0.4}  \Diamond_{0.1} A, \square_{1} \neg A | \Rightarrow }}
{
\infer[\mbox{($\neg$ R)}]{\SEQ{ }{\square_{0.4}  \Diamond_{0.1} A, \neg A | \Rightarrow }}
{
\infer[\mbox{($\square$ R)}]{\SEQ{ A}{\square_{0.4} \Diamond_{0.1} A | \Rightarrow }}
{
\infer[\mbox{($\neg$ R)}]{\SEQ{ A}{ [\Rightarrow \neg \square_{0.1} \neg A ]_{0.4} | \Rightarrow }}
{
\infer[\mbox{($\square$ L sym)}]{\SEQ{ A}{[\square_{0.1} \neg A \Rightarrow ]_{0.4} | \Rightarrow }}
{
\infer[\mbox{($\neg$ L)}]{\SEQ{\neg A, A}{[\Rightarrow ]_{0.4} | \Rightarrow }}
{
 {\SEQ{A}{A, [\Rightarrow ]_{0.4} | \Rightarrow }}
}
}
}
}
}
}
}
}
}
$$

($\neg \square_{1} \neg A, \neg \square_{0.4}  \Diamond_{0.1} A \Rightarrow  | \neg \square_{1} \neg A \Rightarrow $ is equal to $\Diamond_{1} A, \Diamond_{0.4} \sim A \Rightarrow \ |\Diamond_{1} A  \Rightarrow$.)
A similar proof can be carried out for $\square$ other than 0.4, and we do so in the five cases.
We then proceed as follows.

$$
\infer[\mbox{(Classical logic) $^{}$}]{\Rightarrow {\sf M_{0.2}} \sim A \to {\sf \neg M_{1}} A  |  \Rightarrow  \ }
{
\infer[\mbox{($\neg$R) $^{}$}]{{\sf M_{0.2}} \sim A, \Rightarrow {\sf \neg M_{1}} A  |  \Rightarrow  \ }
{
\infer[\mbox{$({\sf M}$L) $^{}$}]{{\sf M_{1}} A, {\sf M_{0.2}} \sim A, \Rightarrow   |  \Rightarrow  \ }
{
\infer[\mbox{(${\sf M}$L) $^{}$}]{ \Diamond_{1} A, {\sf M_{0.2}} \sim A, \Rightarrow    | \Diamond_{1} A \Rightarrow }
{
  \deduce{ \Diamond_{1} A \Rightarrow \ | \Diamond_{1} A, \Diamond_{0.4} \sim A \Rightarrow}
  {
  \deduce{\Diamond_{1} A , \Diamond_{0.1} \sim A \Rightarrow \ | \Diamond_{1} A, \Diamond_{0.3} \sim A \Rightarrow }
  {
  \deduce{\Diamond_{1} A , \Diamond_{0.2} \sim A \Rightarrow \ | \Diamond_{1} A, \Diamond_{0.2} \sim A \Rightarrow }
  {
    \deduce{\Diamond_{1} A , \Diamond_{0.3} \sim A \Rightarrow \ | \Diamond_{1} A, \Diamond_{0.1} \sim A \Rightarrow }
  {
\Diamond_{1} A , \Diamond_{0.4} \sim A \Rightarrow \ |\Diamond_{1} A  \Rightarrow }}
  }}
}
}
}
}
$$


\begin{theorem}[Soundness]\label{soundMSQLM}
If $\GA_1 \Rightarrow \DE_1,\mathcal{T}_1| \GA_2 \Rightarrow \DE_2,\mathcal{T}_2 | \ldots | \GA_{N} \Rightarrow \DE_{N},\mathcal{T}_N$ is provable in {\bf MSQLM}, then $\GA_1 \Rightarrow \DE_1,\mathcal{T}_1| \GA_2 \Rightarrow \DE_2,\mathcal{T}_2 | \ldots | \GA_{N} \Rightarrow \DE_{N},\mathcal{T}_N$ is valid.
\end{theorem}
\begin{proof}
We prove the claim by showing that, for every rule, if the upper MN-sequents are valid, then the lower sequent is also valid.
The cases for rules other than $({\sf M}$L) and $({\sf M}$R) are natural extensions of the cases in \cite{TKMB}, and are therefore omitted.

\begin{enumerate}

\item (${\sf M}$L)

Fix an arbitrary $N_I$-QPM and an embedding $\mathcal{E}$ for the conclusion. 
Suppose that all MN-sequents in $MSQ_1$ of (${\sf M}$L) are valid.
If ${\sf M}_\alpha A$ is false, the conclusion is immediately true because its principal formula occurs on the left-hand side. Suppose therefore that ${\sf M}_\alpha A$ is true. 
For each root sequent, let $\zeta_i = \max ( \{ 0 \} \cup \{\gamma \in J \setminus \{0\} |\ \mathcal{E}(i) \models \Diamond_\gamma A \})$.
Then ${\sf M}_\alpha A$ is true iff $N_I \alpha \leqq \sum_{i=1}^{N_I} \zeta_i$.
Since all $\zeta_i$ and $N_I \alpha$ are multiples of $1/N_J$, we can choose $\gamma_i \in J$ such that $0 \leqq \gamma_i \leqq \zeta_i$ and $N_I \alpha = \sum_{i=1}^{N_I} \gamma_i$.
These $\gamma_i$ determine one of the premises in $MSQ_1$.
We consider the same $\mathcal{E}$ for this MN-sequent.
Then, from $\gamma_i \leqq \zeta_i$, every formula $\Diamond_{(\gamma,i)}A$ in this MN-sequent is true at $\mathcal{E}(i)$.
Because this MN-sequent is valid, ${\SEQ{G_M \{ | \GA_i}{\DE_i, \mathcal{T}_i| \}}}$ is true under $\mathcal{E}$. Therefore, ${\SEQ{G_M \{ | {\sf M_\alpha} A , \GA_i}{\DE_i, \mathcal{T}_i| \}}}$ is also true.
Since this argument holds for any $\mathcal{E}$, ${\SEQ{G_M \{| {\sf M_\alpha} A, \GA_i}{\DE_i, \mathcal{T}_i |\}}}$ is valid.

(For example, let us consider Example 1 of (${\sf M}$L).
The statement that ${\sf M}_{0.2}A$ is true is equivalent to $\Diamond_{0.4} A$ is true at $\mathcal{E}(1)$, or,  $\Diamond_{0.3} A$ is true at $\mathcal{E}(1)$ and  $\Diamond_{0.1} A$ is true at $\mathcal{E}(2)$, or, $\ldots$, or, $\Diamond_{0.4} A$ is true at $\mathcal{E}(2)$. Therefore, 
validity of all upper sequents implies that if ${\sf M}_{0.2}A$ is true, then ${\GA_1 \Rightarrow \DE_1, \mathcal{T}_1 \ |\GA_2 \Rightarrow \DE_2, \mathcal{T}_2}$ is true. 
Therefore, ${\GA_1 \Rightarrow \DE_1, \mathcal{T}_1 \ |{\sf M_{0.2}}A, \GA_2 \Rightarrow \DE_2, \mathcal{T}_2}$ is true.)

\item (${\sf M}$R)

Fix an arbitrary $N_I$-QPM and an embedding $\mathcal{E}$ for the conclusion, and define $\zeta_i$ as above. 

Suppose that all MN-sequents in $MSQ_2$ of (${\sf M}$R) are valid.
If ${\SEQ{G_M \{ | \GA_i}{\DE_i, \mathcal{T}_i | \}}}$ is true under $\mathcal{E}$, then ${\SEQ{G_M \{ | \GA_i}{\DE_i, {\sf M_\alpha} A, \mathcal{T}_i | \}}}$ is clearly true under $\mathcal{E}$.
Suppose that ${\SEQ{G_M \{ | \GA_i}{\DE_i, \mathcal{T}_i | \}}}$ is false under $\mathcal{E}$. 
Assume, for contradiction, that ${\sf M_\alpha} A$ is false. Then, $\sum_{i=1}^{N_I} \zeta_i < N_I \alpha$.
Since both sides are multiples of $1/N_J$, $\sum_{i=1}^{N_I} \zeta_i \leqq N_I \alpha - 1/N_J$.
Hence we can choose $\tau_i \in J$ such that $\zeta_i \leqq \tau_i \leqq 1$ and
$\sum_{i=1}^{N_I} \tau_i = N_I \alpha-1/N_J$:
starting with $\tau_i = \zeta_i$, increase some $\tau_i < 1$ successively by $1/N_J$ until the required sum is reached. Define $S_i = \tau_i + 1/N_J$; when $\tau_i = 1$, interpret $S_i = 1 + 1/N_J$ as the case in which no formula is added to the $i$-th root sequent. Then
$\sum_{i=1}^{N_I} S_i = N_I \alpha + (N_I-1)(1/N_J)$,
so these $S_i$ determine a premise in $MSQ_2$.
We consider the same $\mathcal{E}$ for this MN-sequent.
If $S_i=\gamma_i$, then $\zeta_i < \gamma_i$, and hence $\Diamond_{\gamma_i}A$ is false at $\mathcal{E}(i)$. Since ${\SEQ{G_M \{ | \GA_i}{\DE_i, \mathcal{T}_i | \}}}$ is false under $\mathcal{E}$, this MN-sequent (resulting premise in $MSQ_2$) is also false, contradicting its validity. Therefore ${\sf M_\alpha} A$ must be true, and hence the conclusion is true.

(For example, in the case of Example 2 of (${\sf M}$R), among the patterns where ${\sf M_{0.5}} A$ is false,
the first sequent corresponds to the pattern in which $\Diamond_{0.1} A$ is false at $\mathcal{E}(1)$
and $\Diamond_{0.5} A$ is false at $\mathcal{E}(2)$.
That is, from the definition of $x \models \square_{\gamma} A$, at $\mathcal{E}(1)$, $\Diamond_{\gamma} A$ does not hold for any $\gamma$, and at $\mathcal{E}(2)$, at most $\Diamond_{0.4} A$ can hold. 
In this case, even if $\Diamond_{1} A$ is true at $\mathcal{E}(3)$, the total is $(0 + 0.4 + 1)/3 < 0.5$ and does not reach $0.5$.
Hence, the situation that at least one of these $\Diamond_{(\gamma, i)} A$ is true implies that 
${\sf M_{0.5}} A$ cannot be false in this pattern.
Since this reasoning applies to all upper sequents, ${\sf M_{0.5}} A$ cannot be false.)
\end{enumerate}
\end{proof}

To prove the completeness theorem, we construct a canonical model $(W_c, R_c, V_c, P_c)$ of an unprovable MN-sequent $G_M$.
The canonical model is constructed by a procedure that extends
$G_M$ by successively adding formulas and nodes while preserving unprovability.
This procedure is continued until the MN-sequent is no longer changed by applying any of the steps. 
These steps successively add new formulas or new nodes according to the formulas occurring on the left or right sides of the sequent. The steps are then further applied to these newly added formulas.
At each step, the structural complexity of the formula newly introduced is strictly less than that of the formula on which the step is based.
The complexity of formulas is defined in the usual way for classical operators and for case $\square_{\gamma} A$. For case ${\sf M_{\alpha}} A$, it is defined such that ${\sf M_{\alpha}} A$ has a higher complexity than $\Diamond_{\gamma} A$.
The step of adding a node is performed at most once for each formula.
Therefore, this procedure terminates in finitely many steps.
The sequent obtained at the end of this procedure is denoted by $F(G_M)$.

\noi
The steps of the procedure for constructing  $F(G_M)$:

\begin{enumerate}

\item\label{wedgeL} For the case $\phi \wedge \psi \in \GA$: 
We add $\phi$ and $\psi$ to $\GA$.
This new MN-sequent is also unprovable because of ($\wedge$L).

\item\label{wedgeR} For the case $\phi \wedge \psi \in \DE$: 
Because of ($\wedge$R),
either the MN-sequent obtained by adding $\phi$ to $\DE$ or the sequent obtained by adding $\psi$ to $\DE$ is unprovable.
Of these two, we take the unprovable one as next MN-sequent.

\item\label{negL} For the case $\neg \phi \in \GA$: 
We add $\phi$ to $\DE$.  
This new MN-sequent is unprovable because of ($\neg$L).

\item\label{negR} For the case $\neg \phi \in \DE$: 
We add $\phi$ to $\GA$.  
This new MN-sequent is unprovable because of ($\neg$R).

\item\label{boxL}
If $\square_{\gamma} A \in \GA$, $\gamma \leqq \gamma'$, and $\GA \Rightarrow \DE$ and $\GA' \Rightarrow \DE'$ appear in the MN-sequent in the form of $\GA \Rightarrow \DE, [ \GA' \Rightarrow \DE',\mathcal{T'}]_{\gamma'}, \mathcal{T}$,
then we add $A$ to $\GA'$.
This new MN-sequent is unprovable because of ($\square$L).


\item\label{boxLsym}
If $\square_{\gamma} A \in \GA'$, $\gamma \leqq \gamma'$, and $\GA \Rightarrow \DE$ and $\GA' \Rightarrow \DE'$ appear in the MN-sequent in the form of $\GA \Rightarrow \DE, [ \GA' \Rightarrow \DE',\mathcal{T'}]_{\gamma'}, \mathcal{T}$,
then we add $A$ to $\GA$.
This new MN-sequent is unprovable because of ($\square$L sym).

\item\label{boxLself}
If $\square_{\gamma} A \in \GA$, then we add $A$ to $\GA$.
This new MN-sequent is unprovable because of ($\square$L self).

\item\label{boxR}
If $\square_{\gamma} A \in \DE$ and $(\gamma \neq 1)$, then we add $[\Rightarrow A]_{\gamma}$ to the right side of the sequent where $\DE$ occurs (and $\Rightarrow A$ then becomes a new node).  
This new MN-sequent is unprovable because of the rule ($\square$R).
This step is performed once per occurrence of $\square_{\gamma} A$.


\item\label{boxRself}
If $\square_{1} A \in \DE$, then we add $A$ to $\DE$.
This new MN-sequent is also unprovable because of the rule ($\square$R self).


\item\label{maL} For the case ${\sf M_{\alpha}} A \in \GA_i$

We regard the current MN-sequent as the MN-sequent below (${\sf M}$L). Then at least one of the MN-sequents above (${\sf M}$L) with ${\sf M_{\alpha}} A$ in the left side of $i$-th sequent is unprovable because if all of them are provable, then the current MN-sequent is provable by (${\sf M}$L), and this contradicts the unprovability of current MN-sequent.
Then, we choose one of them and take it as the next MN-sequent; in other words, we add each $\Diamond_{(\gamma, i)} A$ to the root sequents in the current MN-sequent.




\item\label{maR} For the case ${\sf M_{\alpha}} A \in \DE_i$

We regard the current MN-sequent as the MN-sequent below (${\sf M}$R). 
Then, at least one of the MN-sequents above (${\sf M}$R) with ${\sf M_{\alpha}} A$ in the right side of $i$-th sequent is unprovable.
Then, we choose one of them and take it as the next MN-sequent.

\end{enumerate}

The canonical model $(W_c, R_c, V_c, P_c)$ for an unprovable MN-sequent $G_M$ is defined as follows:

\begin{description}


\item $W_c: \ 
\{ \GA \Rightarrow \DE | \GA \Rightarrow \DE \sqsubset F(G_M) \}$

\item $R_c:$ \ 
If $\GA \Rightarrow \DE$ and $\GA' \Rightarrow \DE'$ appear in $F(G_M)$ in the form of 
$\GA \Rightarrow \DE, [ \GA' \Rightarrow \DE',\mathcal{T'}]_{\gamma}, \mathcal{T}$, then $R_c(\GA \Rightarrow \DE, \GA' \Rightarrow \DE') = R_c(\GA' \Rightarrow \DE', \GA \Rightarrow \DE)= \gamma$. $R_c(\GA \Rightarrow \DE, \GA \Rightarrow \DE)= 1$. Otherwise, $R_c(\GA \Rightarrow \DE, \GA' \Rightarrow \DE') = 0$.

\item $V_c:$ \ 
$\GA \Rightarrow \DE \in V_c(p)$ iff $p \in \GA$.

\item $P_c:$ \ 
If $\GA_i \Rightarrow \DE_i$ is a root sequent in the trees
of $F(G_M)$, then $P_c(\GA_i \Rightarrow \DE_i) = 1/N_I$.
Otherwise, $P_c(\GA \Rightarrow \DE) = 0$.

\end{description}

\begin{lemma}\label{TSMBreal} 
$(W_c,R_c, V_c, P_c)$ is a $N_I$-QPM.
\end{lemma}
\begin{proof}
By the definition of $R_c$, every pair of nodes is associated with a single number.
Furthermore, it is only in the case of $x = y$ that $R_c(x,y) =1$ because $\gamma \neq 1$ for $[\ ]_{\gamma}$. By the definition of $R_c$, symmetry follows immediately.
It follows immediately from the definition of $P_c$ that the conditions for $N_I$-QPM are satisfied.
\end{proof}

The embedding $\mathcal{E}_c$ from $G_M$ to $(W_c, R_c, V_c,P_c)$ is defined as follows.
From the configuration of $F(G_M)$, all the nodes (sequents) that existed in $G_M$ also exist in $F(G_M)$ (but with the added formulas).
$\mathcal{E}_c$ maps each node of $G_M$ to the corresponding node of $F(G_M)$.

\begin{lemma}\label{embEc}
$\mathcal{E}_c$ satisfies the embedding conditions.
\end{lemma}
\begin{proof}
We show that $\mathcal{E}_c$ satisfies all four conditions of embedding.
\begin{enumerate}
\item 
This follows directly from the definition.
\item 
From the definition of $P_c$, every root sequent $x$ in $W_c$ satisfies $P_c(x) \neq 0$, and they are the images of the root sequents in $G_M$ under $\mathcal{E}_c$.
\item 
From the construction of $F(G_M)$, the nodes that existed initially remain distinct nodes until the end. Therefore, each node is mapped to a different node.
\item 
From the definition of $R_c$, the values of the edges themselves become the relations.
From the construction of $F(G_M)$, the values of the edges that existed initially remain unchanged until the end. 
Hence, the structure $\GA \Rightarrow \DE, [ \GA' \Rightarrow \DE',\mathcal{T'}]_{\gamma}, \mathcal{T}$ in $G_M$ means $R_c(\mathcal{E}_c(\GA \Rightarrow \DE),\mathcal{E}_c(\GA' \Rightarrow \DE')) = \gamma$, which satisfies $\gamma \leqq \gamma' (= \gamma)$.
    
\end{enumerate}
\end{proof}

\begin{lemma}\label{fCM}
For all $\GA \Rightarrow \DE \in W_c, (\GA \Rightarrow \DE) \models \phi$ if $\phi \in \GA$, and $(\GA \Rightarrow \DE) \not\models \phi$ if $\phi \in \DE$.
\end{lemma}
\begin{proof}
The argument is by induction on the construction of $\phi$, where $\Diamond_\gamma A$ is defined as a subformula of ${\sf M}_\alpha A$.

\begin{enumerate}

\item The cases for classical connectives and $\square_\gamma A$ are natural extensions of the cases in \cite{TKMB}, using Step \ref{wedgeL}, \ref{wedgeR}, \ref{negL}, \ref{negR}, \ref{boxL}, \ref{boxLsym}, \ref{boxLself}, \ref{boxR}, and \ref{boxRself}.







\item Suppose that $ {\sf M}_\alpha A \in \GA$.
From Step \ref{maL}, depending on $\alpha$, multiple $\Diamond_{(\gamma,i)} A$ are added to the left-hand side of each root sequent.
From the inductive hypothesis, these $\Diamond_{(\gamma,i)} A$ are true at the points where each is embedded.
From $\gamma_i \leqq \zeta_i$, $\alpha = (1/N_I)\sum_{i=1}^{N_I} \gamma_i \leqq (1/N_I)\sum_{i=1}^{N_I} \zeta_i$.
Therefore, ${\sf M}_\alpha A$ is true.

\item Suppose that $ {\sf M}_\alpha A \in \DE$.
From Step \ref{maR}, depending on $\alpha$, multiple $\Diamond_{(\gamma,i)} A$ are added to the right-hand side of each root sequent.
From the inductive hypothesis, these $\Diamond_{(\gamma,i)} A$ are false at the points where each is embedded.
From the definition of $x \models \square_{\gamma} A$, for each point $\mathcal{E}_{c}(i)$, any $\Diamond_{\delta} A$ that satisfies $\gamma_i \leqq \delta$ is false. 
Therefore, $\zeta_i \leqq \gamma_i - 1/N_J$.
For points to which $\Diamond_{(\gamma,i)} A$ was not added, we have $\zeta_i \leqq 1 = S_i - 1/N_J$. Therefore, $(1/N_I)\sum_{i=1}^{N_I} \zeta_i \leqq (1/N_I)\sum_{i=1}^{N_I} (S_i - 1/N_J) = \alpha -1/(N_I N_J) < \alpha$, and hence ${\sf M}_\alpha A$ is false.

\end{enumerate}
\end{proof}

\begin{lemma}\label{falseC} 
Let $\GA \Rightarrow \DE \sqsubset G_M$ and $\GA' \Rightarrow \DE' \sqsubset F(G_M)$.
If $\mathcal{E}_c(\GA \Rightarrow \DE) = \GA' \Rightarrow \DE'$ and $\phi \in \GA (\phi \in \DE)$, then $\phi \in \GA' (\phi \in \DE')$. Thus, $G_M$ is false in $(W_c, R_c, V_c, P_c)$ under $\mathcal{E}_c$.
\end{lemma}
\begin{proof}
If $\GA \Rightarrow \DE$ is a node of the original $G_M$, then every $A \in \GA$ remains on the left of $\mathcal{E}_c(\GA \Rightarrow \DE)$, and every $B \in \DE$ remains on the right because all steps do not remove formulas in the composition of $F(G_M)$. Hence, by Lemma \ref{fCM},
$\mathcal{E}_c(\GA \Rightarrow \DE)\models A \quad(A\in\GA)$
and
$\mathcal{E}_c(\GA \Rightarrow \DE)\not\models B \quad(B\in\DE).$
Therefore, for any $\GA \Rightarrow \DE \sqsubset G_M$,  $\GA \Rightarrow \DE$ is false at $\mathcal{E}_c(\GA \Rightarrow \DE)$.
\end{proof}

\begin{theorem}[Completeness]\label{compMSQLM}
If $G_M$ is valid in $N_I$-QPM, then
$G_M$ is provable in {\bf MSQLM}.
\end{theorem}
\begin{proof}
We prove the contrapositive. Assume that $G_M$ is unprovable.
From Lemma \ref{falseC}, there exist $N_I$-QPM and $\mathcal{E}$ such that $G_M$ is false.
Therefore, $G_M$ is not valid.
\end{proof}

The following theorem follows by the standard argument, since a finite canonical model can be constructed and the completeness proof is cut-free.

\begin{theorem}[Cut-elimination]\label{cutMSQLM}
If $G_M$ is provable in {\bf MSQLM}, then
$G_M$ is provable in {\bf MSQLM} without (cut).
\end{theorem}
\begin{proof}
Let {\bf MSQLM}$^{-}$ denote {\bf MSQLM} without (cut). Since none of the steps used in the canonical-model construction relies on (cut), the completeness argument above applies equally to {\bf MSQLM}$^{-}$. Hence every valid MN-sequent is provable in {\bf MSQLM}$^{-}$. 
\end{proof}

The present cut-elimination result is semantic rather than a direct syntactic cut-reduction theorem: the canonical-model construction establishes completeness without using the cut rule. Hence, by soundness and cut-free completeness, every provable MN-sequent has a cut-free proof. Whether a direct syntactic cut-elimination procedure can be given for {\bf MSQLM} remains open.

\begin{theorem}[Decidability]
{\bf MSQLM} is decidable.
\end{theorem}
\begin{proof}
Decidability follows from finite backward proof search in the cut-free calculus. For a given MN-sequent $G_M$, only formulas belonging to the finite closure of the subformulas of $G_M$ and the finitely many threshold modalities determined by $I$ and $J$ can occur. Moreover, each modal formula generates a new node at most once, and the rules (ML) and (MR) have only finitely many premises because $I, J$, and $I_J$ are fixed and finite. Hence backward proof search generates only a finite search tree. By soundness and completeness, $G_M$ is provable exactly when one of these finite searches closes.
\end{proof}




\section{Conclusion and Future Work}

We have proposed QLM, a finite modal framework for reasoning about basic probabilistic aspects of quantum mixed states. The logic separates transition probabilities at the level of pure states from probabilistic aggregation at the level of mixed states. We have also introduced the corresponding proof system {\bf MSQLM} and established its main proof-theoretic properties.

There are several natural directions for future work. One is to combine the present framework with dynamic quantum logic, incorporating notions such as unitary evolution and projection. Another is to extend it toward more concrete quantum phenomena, such as entanglement or decoherence.
Another proof-theoretic direction is to examine whether one can construct a proof system based on ordinary sequents, rather than on hypersequents or related structures.

At the same time, adding more quantum-mechanical structure to a Kripke-style model may obscure the purpose of the abstraction. The aim of QLM is not to replace Hilbert-space quantum logic, but to isolate a specific two-level probabilistic structure in a finite relational setting and make it accessible to modal and proof-theoretic analysis. Further development of such finite modal abstractions seems to be a promising direction.

\clearpage



%


\bibliographystyle{eptcs}
\bibliography{generic}

\clearpage

\end{document}